\documentclass[letterpaper, 10 pt, conference]{ieeeconf}

\usepackage{amsmath}

\usepackage{amsthm}
\usepackage{amsfonts}
\usepackage{bbm} 
\usepackage{graphicx}
\usepackage{subfigure}
\usepackage{cite}
\usepackage{textcomp}
\usepackage{url}
\usepackage{ amssymb }
\usepackage{color}
\usepackage{dsfont}

\newtheorem{defn}{Definition}
\newtheorem{thm}{Theorem}

\newtheorem{prop}{Proposition}
\newtheorem{cor}{Corollary}

\newtheorem{rem}{Remark}
\newtheorem{exam}{Example}
\newtheorem{assump}{Assumption}

\newcommand{\E}{\mathcal{E}}
\newcommand{\EE}{\mathbb{E}}
\newcommand{\V}{\mathbb{V}}
\newcommand{\G}{\mathcal{G}}

\newcommand{\diag}{\mathrm{diag}}

\begin{document}

\title{Network Feedback Passivation of Passivity-Short Multi-Agent Systems}

\author{Miel Sharf and Daniel Zelazo
\thanks{M. Sharf and D. Zelazo are with the Faculty of Aerospace Engineering, Israel Institute of Technology, Haifa, Israel.
    {\tt\small msharf@tx.technion.ac.il, dzelazo@technion.ac.il}.  This work was supported by the German-Israeli Foundation for Scientific Research and Development.}
}

\maketitle
\begin{abstract}
In this paper, we propose a network-optimization framework for the analysis of multi-agent systems with passive-short agents. We consider the known connection between diffusively-coupled maximally equilibrium-independent passive systems, and network optimization, culminating in a pair of dual convex network optimization problems, whose minimizers are exactly the steady-states of the closed-loop system. 
We propose a network-based regularization term to the network optimization problem and show that it results in a network-based feedback using only relative outputs.
We prove that if the average of the passivity indices is positive, then we convexify the problem, passivize the agents, and that steady-states of the augmented system correspond to the minimizers of the regularized network optimization problem. We also suggest a hybrid approach, in which only a subset of agents sense their own output, and show that if the set is nonempty, then we can always achieve the same correspondence as above, regardless of the passivity indices. We demonstrate our results on a traffic model with non-passive agents and limited GNSS reception.
\end{abstract}

\section{Introduction}\label{sec.Intro}
Distributed control has been extensively studied in the last few years, due to its applications in many scientific and engineering fields \cite{Fujita2007,Stan2007,Franchi2011}. One repeatedly used method in cooperative control is the notion of passivity \cite{Bai2011}. It was first introduced in this framework in \cite{Arcak2007} to study group coordination, but was later used in other areas as robotics, biochemical systems and cyber-physical systems \cite{Chopra2006,Scardovi2010,Antsaklis2013}. 

Many variants of passivity have been introduced over the years to tackle problems in cooperative control, including incremental passivity \cite{Pavlov2008} and relaxed co-coercivity \cite{Stan2007,Scardovi2010}. Another important notion is equilibrium-independent passivity (EIP) \cite{Hines2011}, which considers passivity with respect to all steady-state I/O pairs. For EIP systems, the steady-state I/O pairs are related by a single-valued function.
EIP was used to study port-Hamiltonian systems \cite{vanderSchaft2013}, but it does not apply to single integrators and other marginally stable systems.

To tackle this problem, the notion of maximal equilibrium-independent passivity (MEIP) for SISO systems was introduced in \cite{Burger2014}. It also considers passivity with respect to all equilibria, but asks the collection of steady-state input-output pairs to be a maximal monotone relation instead of a function. In \cite{Burger2014,Sharf2018a}, a connection was established between analysis of diffusively-coupled MEIP systems and network optimization theory, culminating in two dual network optimization problems characterizing the steady-states of the diffusively coupled network. This framework was used in \cite{Sharf2017} to solve the synthesis problem, and in \cite{Sharf2018b} to solve a network identification problem.

In practice, many systems are not passive. Their shortage of passivity is usually quantified using passivity indices \cite{Khalil2001,Zhu2014}. We consider a diffusively-coupled network of agents, each having a uniform shortage of passivity across all equilibria. Analysis of these passive short diffusively-coupled systems was tackled in \cite{Jain2018}, and later generalized in \cite{Jain2019}, by regularizing the network optimization problems obtained by the network optimization framework for MEIP systems. However, the solution requires an appropriate loop transformation for each individual agent, which is not applicable in many situations, either because the agents cannot sense their own output, or the agents are not amenable to the network's designer. We propose a different solution here, relying on the network structure to overcome the lack of passivity. Our contributions are as follows:

We propose a Tikhonov-type regularization to the network optimization problem, consisting only of network-level variables. We show that if the sum of the passivity indices over the agents is positive, then the proposed regularization not only convexifies the corresponding optimization problem, but it is equivalent to a \emph{network-only} feedback passivation of the closed-loop system.Furthermore, we propose another Tikhonov-type regularization requiring no assumptions on the passivity indices of the agents, containing both network-level variables, as well as variables belonging to a prescribed set of agents. 
We show that this hybrid regularization term is equivalent to a feedback passivation containing a network-only term, and a loop transformation for each of the agents in the prescribed set. We also show that if the prescribed set is nonempty, then the regularization term convexifies the optimization problem and passivizes the system.

\if(0)
\begin{itemize}
\item We propose a Tikhonov-type regularization to the network optimization problem, consisting only of network-level variables. We show that if the sum of the passivity indices over the agents is positive, then the proposed regularization not only convexifies the corresponding optimization problem, but it is equivalent to a \emph{network-only} feedback passivation of the closed-loop system. 
\item We propose another Tikhonov-type regularization requiring no assumptions on the passivity indices of the agents, containing both network-level variables, as well as variables belonging to a prescribed set of agents. 
We show that this hybrid regularization term is equivalent to a feedback passivation containing a network-only term, and a loop transformation for each of the agents in the prescribed set. We also show that if the prescribed set is nonempty, then the regularization term convexifies the optimization problem and passivizes the system.
\end{itemize}
\fi

The rest of the paper is as follows. Section \ref{sec.NetOpt} reviews MEIP systems and equilibrium-independent passive-short systems. Section \ref{sec.Passivation} formally states and solves the problem.
Section \ref{sec.Simulations} presents a case study, and Section \ref{sec.Conc} concludes the paper.
\paragraph*{Notations}
This work employs basic notions from algebraic graph theory \cite{Godsil2001}.  An undirected graph $\mathcal{G}=(\mathbb{V},\mathbb{E})$ consists of a finite set of vertices $\mathbb{V}$ and edges $\mathbb{E} \subset \mathbb{V} \times \mathbb{V}$.  We denote the edge that has ends $i$ and $j$ in $\mathbb{V}$ by $k=\{i,j\} \in \mathbb{E}$. For each edge $k$, we pick an arbitrary orientation and denote $k=(i,j)$  when $i\in \mathbb{V}$ is the \emph{head} of edge $k$ and $j \in \mathbb{V}$ the \emph{tail}. The incidence matrix of $\mathcal{G}$, denoted $\mathcal{E}\in\mathbb{R}^{|\mathbb{E}|\times|\mathbb{V}|}$, is defined such that for edge $k=(i,j)\in \mathbb{E}$, $[\mathcal{E}]_{ik} =+1$, $[\mathcal{E}]_{jk} =-1$, and $[\mathcal{E}]_{\ell k} =0$ for $\ell \neq i,j$. 
For a convex function $F$, it's dual is also convex and defined by $F^\star(b) = \sup_{a}\{a^Tb - F(a)\}$\cite{Rockafeller1997}. 

\section{Passivity and Network Optimization}\label{sec.NetOpt}

We consider the following dynamical system, defined on a graph $\G = (\V,\EE)$. Namely we consider $|\V|$ agents and $|\EE|$ controllers, having the following state-space models,
\begin{align}
\Sigma_i: \begin{cases}\dot{x}_i = f_i(x_i,u_i) \\ y_i = h_i(x_i,u_i)\end{cases}, \Pi_e: \begin{cases} \dot{\eta}_e = \phi_e(\eta_e,\zeta_e) \\ \mu_e = \psi_e(\eta_e,\zeta_e)\end{cases},
\end{align}
for $i\in\V,\ e\in\EE$. We consider stacked vectors $u = [u_1,...,u_{|\V|}]^T$, and similarly for $y,\zeta$ and $\mu$. The loop is closed by taking $\zeta(t) = \E^Ty(t)$ and $u(t) = -\E\mu(t)$. The closed loop is called a \emph{diffusively-coupled} system, denoted as a triplet $(\G,\Sigma,\Pi)$, and exhibited in Fig. \ref{fig.ClosedLoop}.

\begin{figure}[!t]
\begin{center}
	\subfigure[A diffusively coupled network $(\G,\Sigma,\Pi)$.] {\scalebox{.28}{\includegraphics{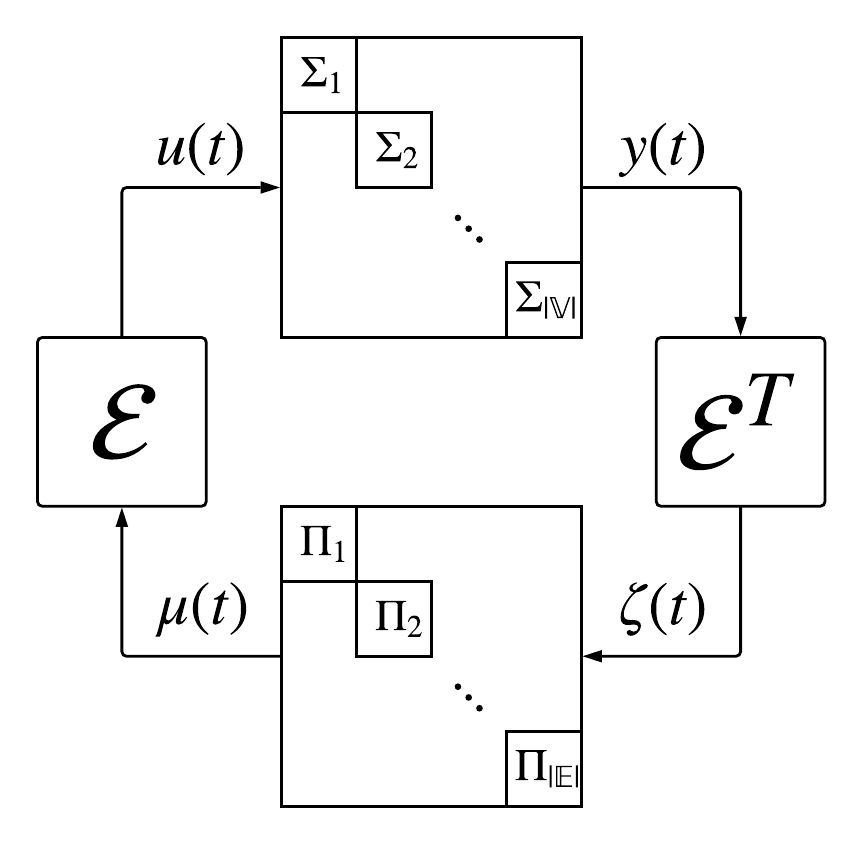}}\label{fig.ClosedLoop}}\hfill
	\subfigure[A general feedback interconnection.] {\scalebox{.28}{\includegraphics{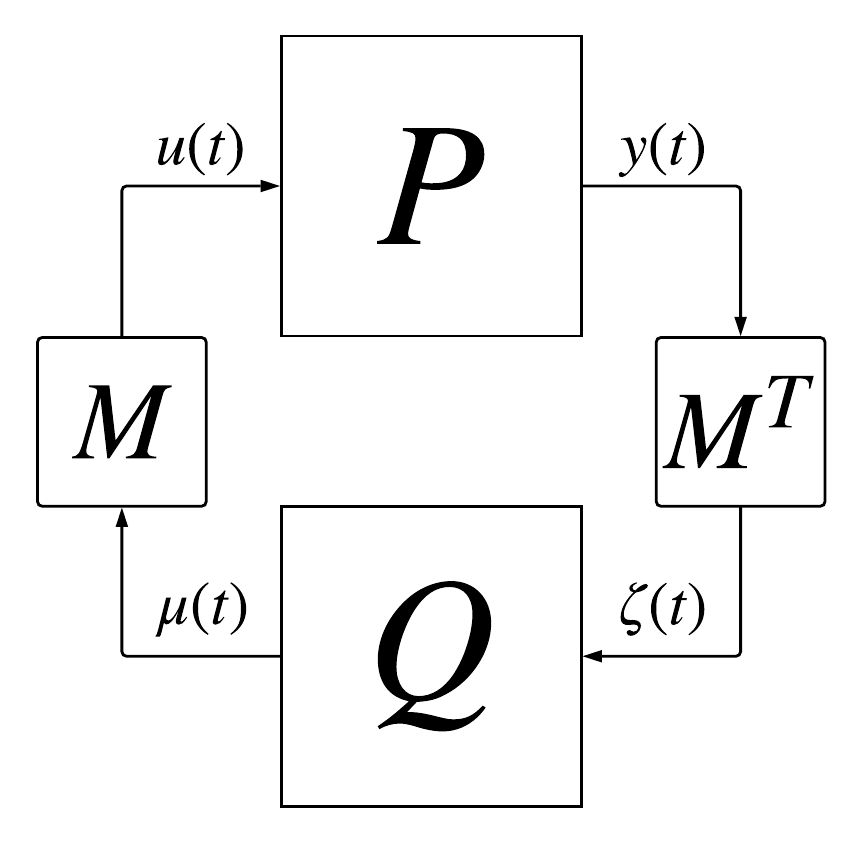}}\label{fig.AbstractFeeback}}
  \caption{Block diagrams of diffusive networks and general feedback systems.}
\end{center}
\vspace{-20pt}
\end{figure}


Our approach to the analysis of the system $(\G,\Sigma,\Pi)$ is based on the notion of equilibrium-independent passivity (EIP) \cite{Hines2011}, and maximal equilibrium-independent passivity (MEIP) \cite{Burger2014}. As the name suggests, these properties require that the system is passive with respect to any equilibrium I/O (I/O) pair. Moreover, they study the collection of I/O pairs of the system. In EIP, we demand that there is a continuous function $k$ mapping steady-state inputs $\mathrm u_{ss}$ to steady-state outputs $\mathrm y_{ss}$. EIP holds for many systems, but it leaves out other important systems, like marginally stable systems, e.g. the single integrator $\dot{x} = u , y=x$, whose steady-state input output-pairs are $\mathrm u=0$ and any $\mathrm y\in\mathbb{R}$. To address this issue, MEIP was proposed in \cite{Burger2014}. In MEIP, we consider the collection $k$ of all steady-state I/O pairs $(\mathrm u_{ss},\mathrm y_{ss})$, called the \emph{steady-state input output relation} of the system.. It gives rise to two set-valued functions, denoted $k$ and $k^{-1}$. If $\mathrm u$ is a steady-state input, and $\mathrm y$ is a steady-state output, we let $k(\mathrm u)$ be the set of all steady-state outputs corresponding to $\mathrm u$, and $k^{-1}(\mathrm y)$ be the set of all steady-state inputs corresponding to $\mathrm y$. We now define MEIP:

\begin{defn}[\hspace{-1pt}\cite{Burger2014}] \label{def.MEIP}
A SISO system is (output-strictly) MEIP if the system is (output-strictly) passive with respect to any steady-state I/O pair $(\mathrm u,\mathrm y)$, and its steady-state relation is maximally monotone.\footnote{For all steady-state I/O pairs $\mathrm{(u_1,y_1),(u_2,y_2)}$, $\mathrm{(u_2-u_1)(y_2-y_1)} \ge 0$ and is not contained in a larger monotone relation.}
\end{defn}

The monotonicity requirement in Definition \ref{def.MEIP} stems from two origins. The first being the fact that the steady-state function $k$ for EIP systems is monotone. The second is that maximally monotone relations are closely tied to convex functions. A theorem by Rockafellar \cite{Rockafellar1966} shows that a maximally monotone relation is the subgradient of a convex function (and vice versa), and the convex function is unique up to an additive constant. 

We now consider a diffusively coupled network $(\G,\Sigma,\Pi)$ with MEIP agents and controllers. We denote the steady-state I/O relations of $\Sigma_i,\Pi_e$ by $k_i,\gamma_e$ respectively, and the stacked versions by $k,\gamma$. By the theorem above, there exists convex functions $K_i,\Gamma_e$ such that $k_i = \partial K_i$ and $\gamma_e = \partial \Gamma_e$, and their sums $K,\Gamma$ satisfy $\partial K = k, \partial \Gamma = \gamma$. In \cite{Sharf2018a}, it is shown that $\mathrm y$ is a steady-state output of the closed-loop system if and only if $0 \in k^{-1}(\mathrm y) + \E\gamma(\E^T\mathrm y)$. Moreover, $\partial K = k, \partial \Gamma = \gamma$, meaning that $k^{-1}(\mathrm y) + \E\gamma(\E^T\mathrm y)$ is exactly the subgradient of $K^\star(\mathrm y) + \Gamma(\E^T \mathrm y)$, where $K^\star(\mathrm y) = \min_{\mathrm u} \{\mathrm y^T\mathrm u - K(\mathrm u)\}$ is the convex dual function of $K$ and $\partial K^\star = k^{-1}$ \cite{Rockafeller1997}. Using convex optimization theory, \cite{Burger2014} proved the following: 
\begin{thm}[\cite{Burger2014}] \label{thm.Analysis}
Consider a diffusively-coupled network $(\G,\Sigma,\Pi)$ with output-strictly MEIP agents and MEIP controllers. Let $K,\Gamma$ be the stacked integral functions for the agents and the controllers, respectively. Then the signals $u(t),y(t),\zeta(t),\mu(t)$ converge to constant steady-states, which are (dual) optimal solutions to the static network optimization problems:
\begin{center}
\begin{tabular}{ c||c }
 \textbf{Optimal Potential Problem} & \textbf{Optimal Flow Problem}  \\ \hline \\
 $ \begin{array}{cl} \underset{y,\zeta}{\min} &K^\star(y) + \Gamma(\zeta)\\
s.t.&\mathcal{E}^Ty = \zeta 
\end{array} $&  $ \begin{array}{cl}\underset{u,\mu}{\min}& K(u) + \Gamma^\star(\mu) \\
s.t. &\mu = -\mathcal{E}u.
\end{array} $ 
\end{tabular}
\end{center}
\end{thm}

\begin{rem} \label{rem.AbstractAnalysisGradientSteadyStates}
The proof of Theorem \ref{thm.Analysis}, as seen in \cite{Sharf2018a}, consists of two parts. The first shows that if there is a steady-state I/O pair $(\mathrm{u,y})$ for the agents $\Sigma$, a steady-state I/O pair $(\zeta,\mu)$ for the controllers $\Pi$, and $\zeta = \E^T\mathrm y, \mathrm u=-\E\mu$, then the closed-loop system converges. This part is based on the output-strict passivity of $\Sigma$ and the passivity of $\Pi$. The second part (for (OPP)) shows that the steady-state equation $0\in k^{-1}(\mathrm y) + \E\gamma(\E^T\mathrm y)$ is equivalent to the minimization of $K^\star(\mathrm y) + \Gamma(\E^T\mathrm y)$. This part is based on the convexity of the integral function $K^\star(\mathrm y) + \Gamma(\E^T\mathrm y)$.

The feedback configuration in Fig. \ref{fig.ClosedLoop} can be thought of more abstractly as the symmetric feedback configuration of two MIMO systems $P$ and $Q$ with the matrix $M$, as shown in Fig. \ref{fig.AbstractFeeback}.
%
This added layer of abstraction, in which we treat the stacked agents and controllers as MIMO dynamical systems and study their I/O steady-state behavior, will be of great importance later. The reason is that $P$, in our case, will be a feedback connection of the agents $\Sigma$ with some network control law, coupling the agents together, and forcing us to consider them as a single, indecomposable system.
\end{rem}


Lastly, we wish to deal with multi-agent systems that have shortage of passivity. We define equilibrium-independent shortage of passivity, as defined in \cite{Jain2018}.
\begin{defn}
The system $\Sigma$ is \emph{equilibrium-independent output-passive short} (EIOPS) if there exist some $\rho < 0$ such that for any steady-state I/O pair $(\mathrm u,\mathrm y)$, there exists a storage function $S(x)$ such that:
\begin{align}\label{eq.OutputPassivityIneq}
\dot{S} \le (u - \mathrm u)^T(y-\mathrm y) - \rho (y - \mathrm y)^T(y - \mathrm y).
\end{align}
\end{defn}
One should note that \eqref{eq.OutputPassivityIneq} also defines passivity (if $\rho = 0$), and output-strict passivity (if $\rho > 0$).

\section{Network Regularization and Passivation}\label{sec.Passivation}
From now on, we fix a collection $\{\Sigma_i\}$ of $n$ agents, and an underlying graph $\G = (\V,\EE)$ on $n$ vertices. 
For the rest of the paper, we also make the following assumption:
\begin{assump} \label{assump.1}
For each $i\in \V$, there is a storage function $S_i$ and some $\rho_i \in \mathbb{R}$ such that the SISO dynamical system $\Sigma_i$ satisfies \eqref{eq.OutputPassivityIneq} for any $(\mathrm {u_i,y_i}) \in k_i$. Moreover, we assume that the inverse steady-state relation $k_i^{-1}$ is a function defined over $\mathbb{R}$. In this case, we can choose an integral function for $k_i^{-1}$ defined by $K_i^\star(\mathrm y) = \int_{\mathrm y_0}^{\mathrm y} k_i^{-1}(\tilde{\mathrm y})\mathrm{d}\tilde{\mathrm y}$. Thus $K_i^\star$ are differentiable and $\nabla K_i^\star = k_i^{-1}$. We denote $K(\mathrm y) = \sum_i K_i(\mathrm y_i)$.
\end{assump}

Suppose one chooses MEIP controllers $\{\Pi_e\}$ over the edges, with steady-state I/O relations $\gamma_e$ and integral functions $\Gamma_e$, and let $\Gamma(\zeta) = \sum_e \Gamma_e(\zeta_e)$. In this case, Theorem \ref{thm.Analysis} does not prove convergence due to lack of passivity, but moreover, the problem (OPP), minimizing $K^\star(\mathrm y) + \Gamma(\zeta)$ such that $\E^T \mathrm y = \zeta$, might not be convex. Indeed, as seen in \cite{Jain2018,Jain2019},  the integral functions $K_i^\star$ might not be convex. In \cite{Jain2018}, a Tikhonov type regularization term of the form $\frac{1}{2}\sum_{i\in\V} \beta_i \mathrm{y}_i^2$, where $\beta_i > 0$ \cite{Boyd2004} was introduced. In turn, this led to the passivation of each agent using the control law $u_i = v_i - \beta_i y_i$, where $v_i$ is some exogenous input, assuming $\beta_i > -\rho_i$. Later, an analysis theorem was established for closed loop with the new, passivized agents, showing that the steady-states of this closed loop networked system correspond to minima of the regularized (OPP), minimizing $K^\star(\mathrm y) + \Gamma(\zeta) + \frac{1}{2}\mathrm y^T\diag(\beta_i)\mathrm y$ with the constraint $\zeta = \E^T \mathrm y$.  For notational convenience, we denote $R=\diag\{\rho_1,\ldots,\rho_{|\V|}\}$.

This method allows for analysis, and later synthesis of passive-short multi-agent systems. However, it requires each agent to implement the control law $u_i = v_i - \beta_i y_i$. This might not be possible in applications for two reasons. First, the agents might not be able to sense their self-output $y_i$, but only relative outputs $y_i - y_j$. This is the case in many formation control problems, or real-life applications for robots in GNSS-deprived areas. Second, the planner of the multi-agent system might not be able to intervene with the agents' dynamics. This is the case in open networks. Thus we strive for a different network regularization term. 

\subsection{Network-Only Regularization and Passivation}

We consider a different Tikhonov-type regularization term, of the form of $\frac{1}{2}  \sum_{e\in \EE} \beta_e\zeta_e^2$, depending only on the network variables $\zeta$. This gives rise to the network-regularized optimal potential problem (NROPP):
\begin{align}\tag{NROPP} \label{NROPP}
\begin{split}
&\hspace{-11pt}\min_{{\rm y}, {\zeta}}K^\star({\rm y}) + \Gamma({\zeta})+\frac{1}{2}\zeta^TB\zeta \;\;{s.t.}~E^T{\rm y} = {\zeta},
\end{split}
\end{align}
where $B=\diag(\beta)=\diag\{\beta_1,\ldots,\beta_{|\EE|}\}$ 
is a design parameter that will be appropriately chosen to make (NROPP) convex. We can consider the cost function of (NROPP) as the sum of two functions - the first is $\Gamma(\zeta)$, which is known to be convex. The second is $K^\star(\mathrm y) + \frac{1}{2} \zeta^T \diag(\beta)\zeta$. Following the notation in \cite{Jain2018} and recalling that $\zeta = \E^T\mathrm{y}$, we denote the latter as
\begin{align} \label{eq.Lambda}
\Lambda_N^\star(\mathrm y) = K^\star(\mathrm y) + 0.5 \mathrm y^T \E\diag(\beta)\E^T \mathrm y.
\end{align}
The following theorem proves that this new, regularized integral function for the agents is induced by a network consensus-type feedback.
\begin{prop} \label{thm.Lambda}
Consider the agents $\Sigma_i$ satisfying Assumption \ref{assump.1}. Let $\Lambda_N^\star$, given by \eqref{eq.Lambda}, be the network-regularized integral function for the agents. Then $\Lambda_N^\star$ is differentiable. Moreover, consider the MIMO system $\tilde{\Sigma}$ given by the parallel interconnection of the agents $\{\Sigma_i\}_{i\in \V}$ with an output-feedback control of the form
\begin{align} \label{eq.feedback}
u = v - \E\diag(\beta)\E^Ty,
\end{align}
with some new exogenous input $v \in \mathbb{R}^n$, and let $\lambda_N$ be its I/O steady-state relation. Then $\lambda_N^{-1}$ is a function, and $\nabla \Lambda_N^\star = \lambda_N^{-1}$. 
\end{prop}

\begin{proof}
The proof is similar to the proof of Theorem~1 in \cite{Jain2018}. $\Lambda_N^\star$ is differentiable as a sum of the differentiable functions $K^\star$ and $\frac{1}{2}\mathrm y^T \E \diag(\beta) \E^T \mathrm y$. Its derivative is given by 
\begin{align} \label{eq.DeriveLambda}
\nabla \Lambda_N^\star(\mathrm y) = k^{-1}(\mathrm y) + \E\diag(\beta)\E^T \mathrm y.
\end{align}
If $(\mathrm{u,y})$ is a steady-state I/O pair for the agents, and we denote $\mathrm v = \nabla \Lambda_N^\star(\mathrm y)$, then \eqref{eq.DeriveLambda} is equivalent to
 \begin{align} \label{eq.DeriveLambda}
\mathrm v = \mathrm u + \E\diag(\beta)\E^T \mathrm y.
\end{align}
Rearranging the terms, we conclude that $(\mathrm {v,y})$ is a steady-state I/O pair for the closed-loop system given by the agents with the network feedback as in \eqref{eq.feedback}. This completes the proof.
\end{proof}

In other words, Proposition \ref{thm.Lambda} gives the following interpretation of (NROPP). It is the optimal potential problem (OPP) for the closed-loop system which is the feedback connection of the controllers $\Pi$ with the augmented agents $\tilde{\Sigma}$, as seen in Fig. \ref{fig.NetworkOnly}. In the spirit of feedback connection of passive systems, and because the controllers $\Pi$ are MEIP, we wish to understand when $\tilde{\Sigma}$ is passive.

\begin{prop} \label{prop.passivity}
Suppose that $R + \E\diag(\beta)\E^T$ is positive-semi definite.\footnote{Recall $R=\diag\{\rho_1,\ldots,\rho_{|\V|}\}$ with $\rho_i$ the passivity index of $\Sigma_i$.} Then $\tilde{\Sigma}$ is passive with respect to any steady-state I/O pair. Moreover, if the matrix is positive-definite, then $\tilde{\Sigma}$ is output-strict passivity.
\end{prop}

\begin{proof}
We take a steady-state I/O pair $(\mathrm v, \mathrm y)$ for $\tilde{\Sigma}$, so that $(\mathrm u,\mathrm y)$ is a steady-state I/O pair of $\Sigma$ where $\mathrm v = \mathrm u + \E\diag(\beta)\E^T\mathrm y$. If $S(x) = \sum_i S(x_i)$ is the sum of the storage functions for the agents $\Sigma_i$, then summing \eqref{eq.OutputPassivityIneq} over the agents gives
$$
\dot{S} \le -(y_i - \mathrm y_i)^TR(y_i - \mathrm y_i) + (y_i - \mathrm y_i)^T(u_i - \mathrm u_i).
$$
Substituting $u_i = v_i - \E\diag(\beta)\E^T\mathrm y$ gives
\begin{align*}
\dot{S} \le &-(y_i - \mathrm y_i)^TR(y_i - \mathrm y_i) + (y_i - \mathrm y_i)^T(v_i - \mathrm v_i) \\ &- (y_i - \mathrm y_i)^T\E\diag(\beta)\E^T(y_i - \mathrm y_i).
\end{align*}
Grouping $R$ and $\E\diag(\beta)\E^T$ completes the proof.
\end{proof}

We conclude the following theorem.
\begin{thm} \label{thm.NROPPAnalysis}
Let $\{\Sigma_i\}_{i\in\V}$ be agents satisfying Assumption \ref{assump.1} with passivity indices $\rho_1,...,\rho_n$. Let $\{\Pi_e\}_{e\in\EE}$ be MEIP controllers with stacked integral function $\Gamma$. Consider the closed-loop system with the controller input $\zeta = \E^T y$ and the control input $u = -\E\mu-\E\diag(\beta)\zeta$. If $R + \E\diag(\beta)\E^T$ is positive-definite, then the closed-loop system converges to a steady-state. Moreover, the steady-state output $\mathrm y$ and $\zeta = \E^T\mathrm y$ are the optimal solutions to the problem (NROPP).
\end{thm}

\begin{figure}[!t]
\begin{center}
	\subfigure[Network only regularization.] {\scalebox{.25}{\includegraphics{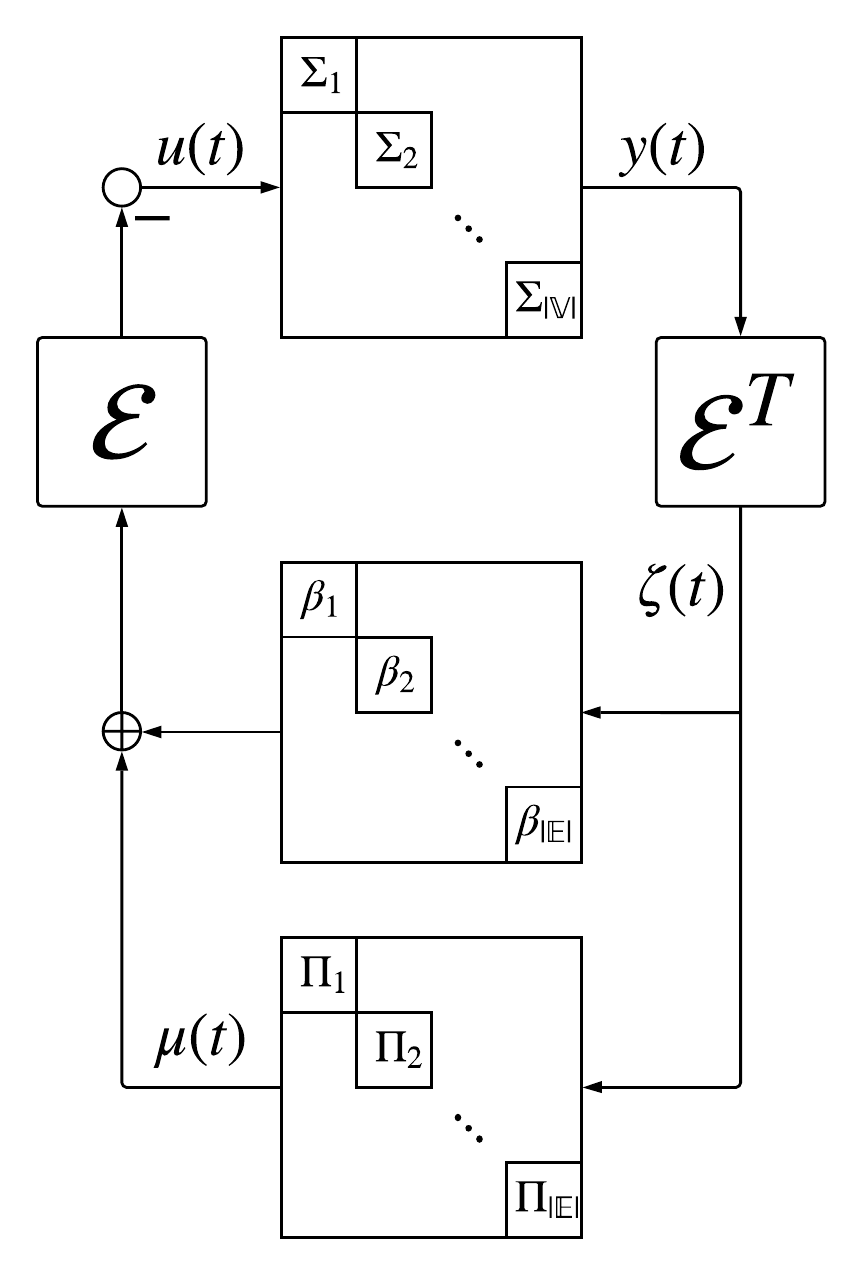}}\label{fig.NetworkOnly}}\hfill
	\subfigure[Hybrid network regularization with one self-regulating agent $\Sigma_{|\V|}$.] {\scalebox{.25}{\includegraphics{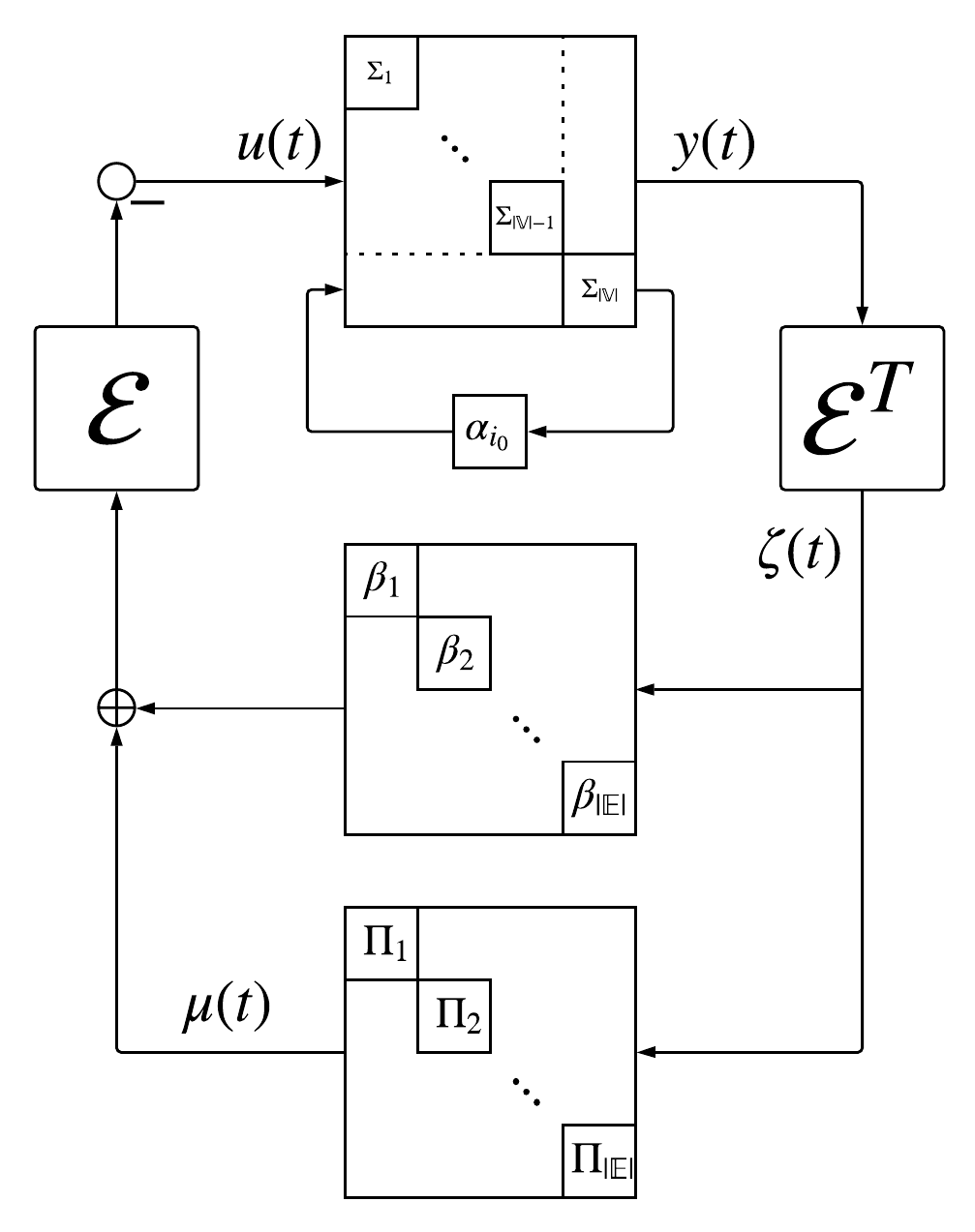}}\label{fig.Hybrid}}
  \caption{Block diagrams of suggested network-based regularization schemes.}
\end{center}
\vspace{-20pt}
\end{figure}
\vspace{-10pt}
\begin{proof}
By the discussion above, the closed-loop system is a feedback connection of the network-regularized agents $\tilde{\Sigma}$, which are output-strictly passive with respect to any steady-state they have, and the controllers $\Pi$, which are MEIP. Moreover, the augmented agents' steady-state I/O relation $\lambda_N^{-1}$ is the gradient of the function $\Lambda_N^\star$. The proof now follows from Theorem \ref{thm.Analysis} and Remark \ref{rem.AbstractAnalysisGradientSteadyStates}.
\end{proof}

We now ask ourselves how to ensure $R + \E\diag(\beta)\E^T$ is positive-definite by appropriately choosing the gains $\beta_e$. To answer that question, we prove the following.
\begin{thm}\label{thm.PassivationMatrix}
Let $\rho_1,...,\rho_{|\V|}$ be any real numbers and assume $\G$ is connected.  There exists some $\beta_1,...,\beta_{|\EE|}$ such that $\diag(\rho)+\E\diag(\beta)\E^T$ is positive definite if and only if $\sum_{i\in\V} \rho_i$ is strictly positive. 
\end{thm}


\begin{proof}
Suppose first that there exist some $\beta_1,...,\beta_{|\EE|}$ such that $X = R + \E\diag(\beta)\E^T$ is positive definite. Then $\mathds{1}_{|\V|}^T X \mathds{1}_{|\V|} > 0$, where $\mathds{1}_{|\V|}$ is the all-one vector. However, $\E^T \mathds{1}_{|\V|} = 0$, so $0<\mathds{1}_{|\V|}^T X \mathds{1}_{|\V|} = \mathds{1}_{|\V|}^TR\mathds{1}_{|\V|} = \sum_i \rho_i$.

As for the other direction, suppose that $\sum_i \rho_i > 0$. We show that if $b$ is large enough, then $R + b\E\E^T$ is positive definite, which will conclude the proof as we can choose $\beta_e = b$. As the matrix in question is symmetric, it's enough to show that for any $x \in \mathbb{R}^{|\V|}$, $x^T(R + b\E\E^T)x \ge 0$.

We can write any vector $x\in\mathbb{R}^{|\V|}$ as $x = \alpha\mathds{1}_{|\V|} + \E z$ for $z\in\mathbb{R}^{|\EE|}$ orthogonal to $\ker(\E)$. The quadratic form is\small
\begin{align*}
x^T(R \hspace{-2pt}+\hspace{-2pt} b\E\E^T)x =& 
\alpha^2\sum_{i}\rho_i \hspace{-2pt}+\hspace{-2pt}  z^T\E^TR(2\alpha\mathds{1}_{|\V|} + \E z) +b||\E^T\E z||^2,
\end{align*}\normalsize
where we use $\E^T\mathds{1}_{|\V|} = 0$. Now, $\E^T\E$ is a positive semi-definite matrix, and $z$ is orthogonal to its kernel, as $\ker(\E^T\E) = \ker(\E)$. Thus $||\E^T\E z|| \ge \lambda_{\mathrm{min},\neq 0}(\E^T\E)||z||$, where $\lambda_{\mathrm{min},\neq 0}(\E^T\E)$ is the minimal non-zero eigenvalue. Moreover, $\E^T\E$ and $\E\E^T$ share nonzero eigenvalues \cite{Horn2012}, hence the minimal nonzero eigenvalue of $\E^T\E$ is $\lambda_2(\G)$, the second eigenvalue of the graph Laplacian. Therefore the quadratic form is bounded from below by{\small{
\begin{align*}
&\alpha^2\sum_{i}\rho_i \hspace{-2pt}+\hspace{-2pt} 2\alpha \hspace{-2pt}z^T\E^TR\mathds{1}_{|\V|} + z^T\E^T\diag(\rho)\E z\hspace{-2pt} +\hspace{-2pt} b\lambda_2^2(\G)||z||^2  \\ &
=\bigg\|\frac{\sqrt{\sum_i \rho_i}}{\sqrt{|\V|}}\alpha\mathds{1}_{|\V|} + \frac{\sqrt{|\V|}}{\sqrt{\sum_i \rho_i}}R\E z\bigg\|^2 \\&+ z^T\bigg(\E^TR\E - \frac{|\V|}{\sum_i \rho_i}\E^TR^2\E + b\lambda_2(\G)^2\mathrm{Id}\bigg)z.
\end{align*} }} \normalsize
The first summand is non-negative, as it is a norm, and the second is non-negative if the symmetric matrix multiplying $z^T$ and $z$ is positive-definite, which is guaranteed if $$b > \frac{\lambda_{\mathrm{max}}\big(\frac{|\V|}{\sum_i \rho_i}\E^T\diag(\rho)^2\E - \E^T\diag(\rho)\E\big)}{\lambda_2(\G)^2} :=\mathrm{\bf b},$$ where $\lambda_{\mathrm{max}}(\cdot)$ is the largest eigenvalue of a matrix. This completes the proof.
\end{proof}
\begin{rem}
Note that if $\G$ is not connected, the result of Theorem \ref{thm.PassivationMatrix} holds if we require that the sum is positive on each connected component.
\end{rem}
\begin{exam}
One might expect that if we only demand positive-semi definiteness in Theorem \ref{thm.PassivationMatrix}, we might be able to accommodate $\sum_{i\in\V_c} \rho_i = 0$. However, this is not the case. Consider a two-agent case with $\rho_1 = 1$ and $\rho_2 = -1$. There is only one edge in the case, $\diag(\rho) + \E\diag(\beta)\E^T=\begin{bmatrix} 1+\beta & -\beta \\ -\beta & -1+\beta\end{bmatrix}$. This matrix can never be positive semi-definite, as its determinant is equal to $-1 < 0$. Thus, the agents cannot be passivized using the network.  
\end{exam}

Theorem \ref{thm.PassivationMatrix} not only gives a good characterization of the diffusively-coupled systems that can be passivized, but also gives a prescription for network passivation, namely have a uniform gain of size $\mathrm{\bf b}+\epsilon$
, for some $\epsilon > 0$, over all edges. However, it shows that not all diffusively-coupled systems satisfying Assumption \ref{assump.1} can be network-passivized. This can be problematic in some applications, e.g. open networks in which the sign of the sum $\sum_{i} \rho_i$ might be volatile, meaning that there might be long periods in which we cannot provide a solution. For that reason, we consider a more general approach in the next subsection.

\subsection{Hybrid Approaches for Regularization and Passivation}
We return to the problem of regularizing the non-convex network optimization problem (OPP). We consider a quadratic regularization term of the form $\frac{1}{2}\zeta^T\diag(\beta)\zeta$ as before, but add another Tikhonov-type regularization type in the spirit of \cite{Jain2018}, $\frac{1}{2}\mathrm y^T \diag(\alpha) y$. Namely, we consider the Hybrid-Regularized Optimal Potential Problem (HROPP),
\begin{align}\tag{HROPP} \label{HROPP}
\begin{split}
&\min_{{\rm y}, {\zeta}}~~K^\star({\rm y}) + \Gamma({\zeta})+\frac{1}{2}\zeta^T\diag(\beta)\zeta + \frac{1}{2}\mathrm y^T\diag(\alpha)\mathrm y\\
&{s.t.}~~~~E^T{\rm y} = {\zeta}.
\end{split}
\end{align}
As we'll see, unlike in \cite{Jain2018}, the vector $\alpha = [\alpha_1,...,\alpha_{|\V|}]^T$ can be very sparse. Namely, we can prove the following.
\begin{thm} \label{thm.HybridLambda}
Consider the agents $\Sigma_i$ satisfying Assumption \ref{assump.1}. Let $\Lambda_H(\mathrm y) =K^\star({\rm y}) + \frac{1}{2}\zeta^T\diag(\beta)\zeta + \frac{1}{2}\mathrm y^T\diag(\alpha)\mathrm y $. Then $\Lambda_H^\star$ is differentiable. Moreover, consider the MIMO system $\tilde{\Sigma}$ given by the agents with the output-feedback control
\begin{align} \label{eq.feedback}
u = v - \E\diag(\beta)\E^Ty - \diag(\alpha)y,
\end{align}
with some new exogenous input $v \in \mathbb{R}^n$, and let $\lambda_H$ be its I/O steady-state relation. Then $\lambda_H^{-1}$ is a function, and $\nabla \Lambda_H^\star = \lambda_H^{-1}$. 
\end{thm}
\begin{proof}
Similar to the proof of Proposition \ref{thm.Lambda}.
\end{proof}
\begin{thm}\label{thm.ConvergenceLambda}
Let $\{\Sigma_i\}_{i\in\V}$ be agents satisfying Assumption \ref{assump.1}. Let $\{\Pi_e\}_{e\in\EE}$ be MEIP controllers with stacked integral function $\Gamma$. Consider the closed-loop system with the controller input $\zeta = \E^T y$ and the control input $u = -\E\mu-\E\diag(\beta)\zeta-\diag(\alpha)y$. If the matrix $\diag(\rho+\alpha) + \E\diag(\beta)\E^T$ is positive-definite, then the closed-loop system converges. Moreover, the steady-state output $\mathrm y$ and $\zeta = \E^T\mathrm y$ are the optimal solutions to the problem (HROPP).
\end{thm}
\begin{proof}
Similar to the proof of Theorem \ref{thm.NROPPAnalysis}. 
\end{proof}

\begin{cor} (Almost Network-Only Regularization)
Let $\{\Sigma_i\}_{i\in\V}$ be agents satisfying Assumption \ref{assump.1}, and suppose that the graph $\G$ is connected. Let $\V_{sr}\subseteq \V$ be any nonempty subset of the agents. Let $\{\Pi_e\}_{e\in\EE}$ be MEIP controllers with stacked integral function $\Gamma$. Consider the closed-loop system with the controller input $\zeta = \E^T y$ and the control input $u = -\E\mu-\E\diag(\beta)\zeta-\diag(\alpha)y$. Then there exist vectors $\alpha \in \mathbb{R}^{|\V|}, \beta \in \mathbb{R}^{|\EE|}$ such that:
\begin{itemize}
\item[i)] For any vertex $i\not\in\V_{sr}, \alpha_i = 0$.
\item[ii)] The closed-loop system converges to a steady-state.
\item[iii)] The steady-state output $\mathrm y$ and $\zeta = \E^T\mathrm y$ are the optimal solutions to the optimization problem (HROPP).
\end{itemize}
\end{cor}
An example of the closed-loop system for a single self-regulating agent $\Sigma_{|\V|}$ can be seen in Fig. \ref{fig.Hybrid}.
\begin{proof}
By Theorem \ref{thm.ConvergenceLambda}, it's enough to find some $\alpha,\beta$ satisfying the first condition such that $\diag(\rho+\alpha) + \E\diag(\beta)\E^T$ is positive-definite. Fixing $\alpha$, Theorem \ref{thm.PassivationMatrix} implies that there is some $\beta$ such that $\diag(\rho+\alpha) + \E\diag(\beta)\E^T$ is positive-definite if and only if $\sum_i (\rho_i + \alpha_i) > 0$, or $\sum_i \alpha_i > -\sum_i \rho_i$. Taking any $i_0 \in \V_{sr}$ and choosing $\alpha_i = 1-\sum_i \rho_i$ for $i=i_0$, and $\alpha_i=0$ otherwise, finishes the proof.
\end{proof}

\begin{rem}
The set $\V_{sr}$ in the theorem can be thought of the set of vertices that can sense their own output, and are amenable to the network designer (i.e., self-regularizable agents). 
The theorem shows the strength of the hybrid approach for regularization of (OPP). we can choose almost all $\alpha_i$-s to be zero - namely one agent is enough. 
In practice, this solution is less restrictive than the one offered in \cite{Jain2018}. 
\end{rem}
\if(0)
Let us conclude all the results into one, concise theorem:
\textcolor{red}{(for space, we may just eliminate this theorem since it is "review" in the end)}
\begin{thm} \label{thm.AnalysisHROPP}
Let $\{\Sigma_i\}_{i\in\V}$ be agents satisfying Assumption \ref{assump.1}. Let $\{\Pi_e\}_{e\in\EE}$ be MEIP controllers with stacked integral function $\Gamma$. Consider the closed-loop system with the controller input $\zeta = \E^T y$ and the control input $u = -\E\mu-b\E\zeta-\diag(\alpha)y$. Suppose that $\sum_{i\in\V} \alpha_i > -\sum_{i\in\V} \rho_i$, and that $b > \frac{\lambda_{\mathrm{max}}\big(\frac{|\V|}{\sum_i \rho_i+\alpha_i}\E^T\diag(\rho+\alpha)^2\E - \E^T\diag(\rho+\alpha)\E\big)}{\lambda_2(\G)^2}$. Then the closed-loop system converges. Moreover, the steady-state output $\mathrm y$ and $\zeta = \E^T\mathrm y$ are the optimal solutions to the problem (HROPP):.
\begin{align}\tag{HROPP}
\begin{split}
&\min_{{\rm y}, {\zeta}}~~K^\star({\rm y}) + \Gamma({\zeta})+\frac{1}{2}b\zeta^T\zeta + \frac{1}{2}\mathrm y^T\diag(\alpha)\mathrm y\\
&{s.t.}~~~~E^T{\rm y} = {\zeta}.
\end{split}
\end{align}
\end{thm}
\fi
\section{Case Studies}\label{sec.Simulations}
Consider the traffic dynamics model proposed in \cite{Bando1995}, in which vehicles adjust their velocity $x_i$ according to the equation
$\dot{x}_i = \kappa_i(V_i(\Delta p) - x_i)$,
 where $\kappa_i>0$ is a constant representing the sensitivity of the $i$-th driver, and 
\begin{align}
V_i(\Delta p) = V_i^0 + V_i^1\sum_{j\sim i} \tanh(p_j - p_i),
\end{align}
is the adjustment, where $V_i^0$ are the preferred velocities, and $V_i^1$ are the ``sensitivity coefficients." This model was studied in \cite{Burger2014}, where it was shown that it can inhibit a clustering phenomenon. In \cite{Jain2018}, the case of $\kappa_i < 0$, attributed to drowsy driving, was studied. There, a self-gain-feedback was applied for each agent, resulting in a network of MEIP agents.

Consider a case where only some agents know their own velocity (e.g., by a GNSS measurement). Thus, agents which have no GNSS reception cannot implement the regularization of (OPP), or the self-feedback loops, proposed in \cite{Jain2018}. Instead, we opt for the hybrid regularization suggested above.

The model is a diffusively coupled network with the agents $\Sigma_i: \dot{x_i} = \kappa_i(-x_i +V_i^0+V_i^1u),\ y_i = x_i$ and the controllers $\Pi_e:\ \dot{\eta}_e = \zeta_e,\ \mu_e = \tanh(\eta_e)$. The agents are EIOPS if $V_i^1\kappa_i > 0$, with $\rho_i = \kappa_i$, so $\kappa_i > 0$ corresponds to output-strict MEIP. We suppose that only agent $i_0$ has GNSS reception, so we implement a correction term of the form $\alpha_{i_0} y_{i_0}^2 + \beta\zeta^T\zeta$ to (OPP), giving us (HROPP). The new control law is $u(t) = -\alpha_{i_0} x_{i_0}e_{i_0} - \beta \E\E^T x - V(\Delta p)$ where $V(\Delta p) = [V_1(\Delta p),\cdots,V_n(\Delta p)]^T$, and $e_i$ is the $i$-th standard basis vector. Only the states $x_{i_0}$ and $\E^T x$ are used in the control law, meaning that no agent but $i_0$ is required to know its velocity in a global frame of reference, but only positions and velocities relative to its neighbors. 

To illustrate this, we consider a network of $n=100$ agents, all connected to each other, with parameters $\kappa_i$ randomly chosen either as $-1$ (w.p. $1/3$) and $1$ (w.p. $2/3$). Moreover, the parameters $V_i^0$ were chosen as a Gaussian mixture model, with half of the samples having mean $20$ and standard deviation $15$, and half having mean $120$ and standard deviation $15$. Lastly, $V_i^1$ where chosen as $0.8\kappa_i$. In \cite{Burger2014}, it is shown that (OPP) in this case is given by
\begin{align*}
\min_{y,\zeta}\quad &\sum_i \frac{1}{2V_i^1}(y_i-V_i^0)^2 + \sum_e |\zeta_e|& \;
\mathrm{s.t.}\quad &\zeta = \E^Ty,&
\end{align*}
meaning that (HROPP) is given by:
\begin{align*}
\min_{y,\zeta}\quad &\sum_i \frac{1}{2V_i^1}(y_i-V_i^0)^2 + \sum_e |\zeta_e| + \alpha_{i_0} y_{i_0}^2 + \beta\sum_e \zeta_e^2&\\
\mathrm{s.t.}\quad &\zeta = \E^Ty.&
\end{align*}

Here, the sum $\sum_i \rho_i=\sum_i \kappa_i$ is positive, so we use the network-regularization method, choosing $\alpha_{i_0} = 0$, and (HROPP) reduces to (NROPP). Choosing $\beta = \mathrm{\bf b} + \epsilon$, we apply Theorems \ref{thm.ConvergenceLambda} and \ref{thm.PassivationMatrix} to conclude that the system converges, and find its steady-state limit.
We plot the trajectories of the system in Fig. \ref{fig.CaseStudyTraj}, as well as the minimizer of (NROPP) in Fig. \ref{fig.CaseStudyTheoreticalLimit}. It is easily seen that the steady-state value of the system matches the forecast, namely the minimizer of (NROPP). It should be noted that we get a clustering phenomenon, as noted in \cite{Burger2014}. However, the clustering is far less refined than one expects given the simulation results presented in \cite{Burger2014}, and the Gaussian mixture model chosen for $V_i^0$. This is due to the ever-going consensus feedback $-\beta\E\E^T x$ appearing in $u$, which does not only passivizes the system, but also forces the trajectories closer to a consensus. 

\begin{figure}[!t]
\begin{center}
	\subfigure[Vehicles' trajectories under network-only regularization.] {\scalebox{.30}{\includegraphics{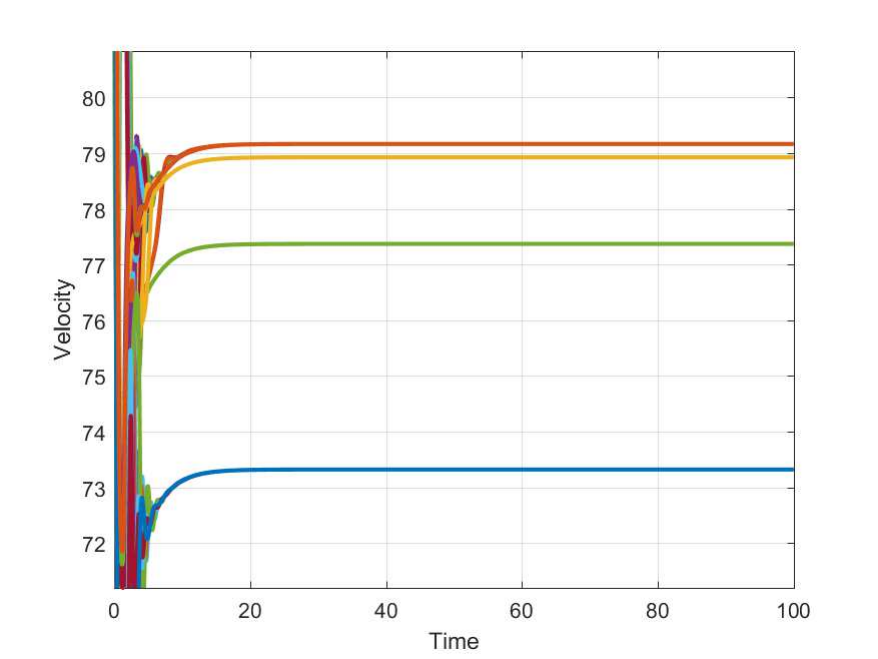}}\label{fig.CaseStudyTraj}}\hfill
	\hfill\subfigure[Asymptotic behaviour predicted by (NROPP).] {\scalebox{.30}{\includegraphics{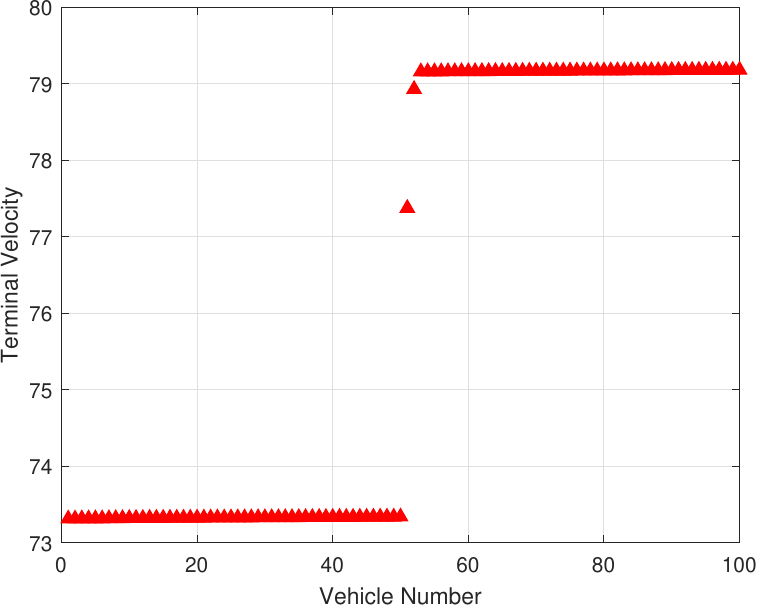}}\label{fig.CaseStudyTheoreticalLimit}}
  \caption{Traffic control model.}
\vspace{-8pt}
\end{center}
\vspace{-20pt}
\end{figure}

\section{Conclusion} \label{sec.Conc}
We considered a diffusively-coupled network of equilibrium-independent output-passive-short (EIOPS) agents. 
As the systems are not MEIP, the cost function of the associated network optimization problem (OPP) need not be convex. We proposed a regularization term, which is based only on the network-level variables $\zeta$, and proved that it corresponds to applying a network-only feedback term on the agents. In turn, we showed that if $\sum_i \rho_i$ is positive, where $\rho_i$ is the passivity parameter of the $i$-th agent,  then we successfully convexify the cost function of (OPP), and that steady-state outputs of the new closed-loop dynamical system correspond to minimizers of the regularized constrained minimization problem, (NROPP). We also suggested a hybrid approach, in which we try and regularize (OPP) both with network-level variables $\zeta$ and a subset of the agent outputs. 
This implies that other agent need not measure their own output, nor self-regulate. 
We showed that (OPP) can always be convexified using this term, and that steady-state outputs of the new closed-loop system correspond to minimizers of the regularized constrained minimization problem, (HROPP). 
Future research can try and find a more refined network-based regularization term, in which different edges are assigned different gains. 
\bibliographystyle{ieeetr}
\bibliography{main}

\end{document}